\documentclass[aps,prl,reprint,amsmath,amssymb,floatfix,superscriptaddress]{revtex4-2}

\usepackage{graphicx}
\usepackage{amsthm}
\usepackage{bm}
\usepackage[colorlinks=true,linkcolor=blue,citecolor=blue,urlcolor=blue]{hyperref}

\theoremstyle{plain}
\newtheorem{theorem}{Theorem}
\newtheorem{lemma}{Lemma}
\newtheorem{corollary}{Corollary}

\newcommand{\eeq}{\epsilon_{\mathrm{eq}}}
\newcommand{\Th}{\Theta}
\newcommand{\Pkj}{P_{K,j}}
\newcommand{\Zd}{Z^{\downarrow}}
\newcommand{\Geq}{G^{\mathrm{eq}}}

\newcommand{\Appendix}{%
  \onecolumngrid\vspace{10pt}
  \begin{center}\rule{0.42\textwidth}{0.5pt}\\[5pt]
    {\bfseries Appendix}\end{center}
  \vspace{2pt}\twocolumngrid
  \setcounter{equation}{0}\renewcommand{\theequation}{A\arabic{equation}}%
  \renewcommand{\theHequation}{A\arabic{equation}}}
\newcounter{appx}
\renewcommand{\theappx}{\Alph{appx}}

\newcommand{\app}[2]{%
  \refstepcounter{appx}\label{#1}%
  \smallskip\noindent\textit{Appendix~\theappx: #2}---}

\begin{document}

\title{Where Energy Is Spent Sets the Depth of Kinetic Proofreading}

\author{U\u{g}ur \c{C}etiner}
\email{ucetiner@ku.edu.tr}
\affiliation{Department of Physics, Koç University\\
Rumelifeneri Yolu 34450 Sarıyer, İstanbul, Türkiye\\}

\begin{abstract}
Kinetic proofreading buys accuracy with energy. Where the energy is spent, we show, caps how much accuracy it can buy. Drive layer $j$ of a $K$-checkpoint cascade and the error never falls below $\eeq^{K-j+1}$, where $\eeq$ is the error the same network reaches at equilibrium. The exponent is one plus the first Betti number of the downstream
subgraph, which counts its independent cycles. Driving the first layer can recruit all $K$ checkpoints, whereas driving the final layer restricts the best possible scaling to $\eeq^2$, independently of network depth. We also derive the exact dependence on driving strength and show that the error either decreases monotonically or has a unique global minimum, beyond which stronger driving worsens accuracy. Drive placement therefore imposes a topological limit on nonequilibrium error correction.
\end{abstract}

\maketitle
\begin{figure*}[t]
  \centering
  \includegraphics[width=\linewidth]{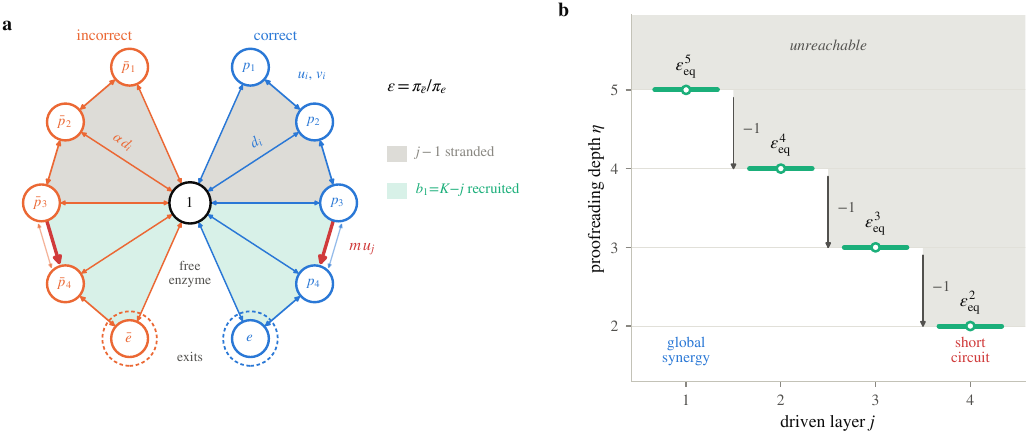}
  \caption{\textbf{Energy expenditure as a topological gate.} \textbf{(a)}~The core butterfly graph: two wings glued at the free-enzyme vertex $1$, here for $K=5$ and drive at layer $j=3$. Each spoke is a reversible binding/fall-off pair and discrimination enters through the fall-off rate, $d_i$ on the correct wing and $\alpha d_i$ on the incorrect one. The exits are $e=p_K$, $\bar e=\bar p_K$. Each triangular cell $\{1,p_r,p_{r+1}\}$ is one independent cycle, counted per wing. The $K-j$ green cells, at or downstream of the driven pair, can be
  recruited for error correction while the $j-1$ gray ones cannot, regardless of how discriminatory they are. \textbf{(b)}~The staircase: the effective proofreading depth $\eta=\ln(\inf\epsilon)/\ln\eeq$ falls by one for each layer the drive moves downstream, from global synergy ($\eta=K$) to the short circuit ($\eta=2$). Only the position of the drive changes along this axis.}
  \label{fig:law}
\end{figure*}

Biological copying is accurate to a degree that equilibrium binding cannot easily explain. Pauling noted that discrimination between a correct and an incorrect substrate by binding free energy alone is limited by the Boltzmann factor of their free-energy gap \cite{Pauling1957}. Valine is misincorporated at isoleucine positions in rabbit hemoglobin at only a few parts in $10^{4}$ \cite{Loftfield1972}. Matching this fidelity at equilibrium would require roughly $8\,k_{\mathrm{B}}T$ of discrimination, far more than the small structural difference between valine and isoleucine can plausibly supply \cite{Hopfield1974}.

Hopfield showed how nonequilibrium driving can overcome this limitation \cite{Hopfield1974}: in the strongly driven regime, an activated intermediate provides a second discriminatory opportunity, reducing the error from $\eeq$ toward $\eeq^{2}$, where $\eeq<1$ is the equilibrium error. Ninio reached the same conclusion independently \cite{Ninio1975}. Subsequent work has characterized kinetic proofreading across network architectures and quantified the tradeoffs among error, speed, and dissipation \cite{Bennett1979,Ehrenberg1980,Murugan2012,Sartori2013,Murugan2014,Sartori2015,Rao2015,Banerjee2017,Wong2018,cetiner2023universal}. The mechanism has been implicated in processes ranging from aminoacyl-tRNA synthesis to T-cell antigen discrimination \cite{Hopfield1976,McKeithan1995}.

Squaring is not the ceiling. A network with $K$ discriminatory checkpoints carries $K-1$ independent cycles per wing (Fig.~\ref{fig:law}(a)), and we have shown previously that breaking detailed balance at a \emph{single} transition (per wing) can, in the best case, push the error as low as $\eeq^{K}$ \cite{Cetiner2022}. That is, one transition can unlock the full discriminatory potential of all $K$ checkpoints, which we call \emph{global synergy}. But ``a single transition'' is not one choice---it is $K-1$ choices---and Hopfield's picture, built on a single cycle, gives no reason to expect them to differ.

They differ maximally. We take the undriven network to obey detailed balance and break it by multiplying the symmetry-related pair $p_j\to p_{j+1}$ and $\bar p_j\to\bar p_{j+1}$ (Fig.~\ref{fig:law}(a)) by a common factor $m>1$. Multiplying by $m>1$ can be read as coupling these two transitions, and only these, to ATP hydrolysis \cite{Hopfield1974}. Here we prove that the steady-state error of a core butterfly network with $K$ checkpoints then obeys $\epsilon>\eeq^{\,K-j+1}$, for every admissible choice of rate constants and every finite $m>1$, and that this bound is sharp.

The exponent has a purely topological reading. Writing $\eta$ for the \emph{effective proofreading depth} [Fig.~\ref{fig:law}(b)],
\begin{equation}
  \eta\;\equiv\;\frac{\ln\left(\inf\epsilon\right)}{\ln\eeq}
  \;=\;1+K-j\;=\;1+b_1\!\left(C_j^{\downarrow}\right),
  \label{eq:depth}
\end{equation}
where $C_r^{\downarrow}$ denotes the induced subgraph on $\{1,p_r,\ldots,p_K\}$, one wing with the vertices $p_1,\ldots,p_{r-1}$ deleted and the root retained; $C_j^{\downarrow}$ is thus the part of the wing from the driven edge onward, and its first Betti number $b_1(C_j^{\downarrow})$ counts the independent cycles lying at or downstream of the drive. Equilibrium alone supplies one factor of $\eeq$, and each of those cycles supplies at most one
more. The $j-1$ upstream cycles contribute nothing to the exponent: they are physically present and every bit as discriminatory, yet they cannot increase the effective proofreading depth.

Nonequilibrium accuracy is usually posed as a currency problem: how much error can be removed for how much dissipation, and how fast \cite{Bennett1979,Murugan2012,Sartori2015,Rao2015,Lan2012,Banerjee2017}. Eq.~\eqref{eq:depth} identifies a resource spent before any currency changes hands: two networks with identical binding free energies, identical undriven rates, and the same driving affinity can differ by orders of magnitude in the fidelity
they can \emph{reach}, purely because the driven transition sits at a different layer. At $j=K-1$ the cascade suffers an \emph{information-theoretic short circuit}, $\eta=2$ however deep it is. Position does not trade against
dissipation but rather caps what dissipation can achieve.

\textit{Framework.}---We model the system as a continuous-time Markov process on a finite state space, encoded as a graph $G$ whose vertices are the states and whose directed edges carry labels $\ell(i\to j)>0$ read as transition rates \cite{Gunawardena2012,Mirzaev2013}. Every edge is accompanied by its reverse and $G$ is strongly connected, so a unique steady state exists. The occupation probabilities $x_i(t)=\Pr\{X(t)=i\}$, collected into the column vector $\bm{x}(t)$, obey the master equation $\dot{\bm{x}}=\mathcal{L}(G)\,\bm{x}$, whose generator is the graph Laplacian $\mathcal{L}(G)$ defined as
\begin{equation}
  \bigl[\mathcal{L}(G)\bigr]_{ji}=
  \begin{cases}
    \ell(i\to j), & j\neq i,\\[1pt]
    -\sum_{k\neq i}\ell(i\to k), & j=i .
  \end{cases}
  \label{eq:laplacian}
\end{equation}
The steady state $\pi$ is the unique normalized vector in $\ker\mathcal{L}(G)$.

A \emph{simple path} $P:i_1\to\cdots\to i_k$ is a sequence of edges whose vertices are all distinct; a \emph{cycle} $\gamma$ is a closed path whose only repeated vertex is its start. Every path below is simple. To each path we assign the \emph{action} \cite{Lebowitz1999}
\begin{equation}
  S(P)=\sum_{r=1}^{k-1}\ln\frac{\ell(i_r\to i_{r+1})}{\ell(i_{r+1}\to i_r)}.
  \label{eq:action}
\end{equation}
If $P$ is a cycle, $S(P)$ is called the \emph{affinity} of that cycle \cite{Seifert2008,Schnakenberg1976}.

The steady state satisfies detailed balance precisely when every cycle has vanishing affinity, $S(\gamma)=0$. This equivalence is known as the \emph{cycle condition}. Its consequence is that $S$ is \emph{path independent}: any two simple paths with the same endpoints carry the same action. Writing $S_{\mathrm{eq}}(i,1)$ for that common value along any path from $i$ to a fixed reference vertex $1$,
\begin{equation}
  \pi_i\;\propto\;e^{-S_{\mathrm{eq}}(i,1)},
  \label{eq:boltzmann}
\end{equation}
which is the analog of the Boltzmann distribution for a Markov process on a discrete state space.

Away from equilibrium, some cycle has $S(\gamma)\neq0$, path independence is lost, and Eq.~\eqref{eq:boltzmann} is no longer well posed: two paths from $i$ to $1$ return different actions, and nothing selects between them. The way out is not to select but to average. A \emph{spanning tree of $G$ rooted at $1$} is a subgraph containing all vertices, free of cycles, in which every vertex but the root has exactly one outgoing edge [Fig.~\ref{fig:objects}(a,b), Appendix]. In such a tree $T$ every vertex $i$ reaches the root along a unique simple path, which we write $T_i$. Let \(\Th_1(G)\) be the set of these trees and \(w(T)=\prod_{i\to j\in T}\ell(i\to j)\) their weights. Normalizing these weights defines the \emph{arboreal distribution}, under which \cite{Cetiner2022}
\begin{equation}
  \pi_i\;\propto\;\bigl\langle e^{-S(T_i)}\bigr\rangle,
  \qquad
  \langle\,\cdot\,\rangle\equiv
  \sum_{T\in\Th_1(G)}
  \frac{w(T)}{w\bigl(\Th_1(G)\bigr)}(\,\cdot\,),
  \label{eq:arboreal}
\end{equation}
where \(w(\Th_1(G))=\sum_{T\in\Th_1(G)}w(T)\) and the angular brackets represent an average over the arboreal distribution.

Three things matter here. (i) Eq.~\eqref{eq:arboreal} is exact arbitrarily far from equilibrium. (ii) At equilibrium every tree returns the same action and the average collapses to Eq.~\eqref{eq:boltzmann}. (iii) Most importantly, the viewpoint is path-centric: for fixed endpoints the number of distinct path actions is set by the number of energetic transitions, not by graph size, which keeps the theorem below tractable as the graph grows (Appendix~\ref{app:edge}). 

\textit{Network and drive.}---The core butterfly graph $G$ [Fig.~\ref{fig:law}(a)] consists of a correct wing on $\{1,p_1,\ldots,p_K\}$ and its mirror image, joined only at the free-enzyme vertex $1$. Within a wing $p_i\to p_{i+1}$ and $p_{i+1}\to p_i$ carry rates $u_i,v_i$ and every $p_i$ can bind from and fall off to $1$; the exits are $e=p_K$, $\bar e=\bar p_K$, and the error is $\epsilon=\pi_{\bar e}/\pi_e$. Discrimination enters only through the fall-off rates,
\begin{equation}
  \ell(\bar p_i\to1)=\alpha\,\ell(p_i\to1)=\alpha\, d_i,
  \qquad \alpha>1,
  \label{eq:alpha}
\end{equation}
the two wings sharing their binding rates, $\ell(1\to\bar p_i)=\ell(1\to p_i)$, and their internal rates $u_i,v_i$. 
We call any strictly positive choice of rates satisfying detailed balance before driving \emph{admissible}, and write $\Geq$ for the resulting undriven graph. Eqs.~\eqref{eq:boltzmann} and \eqref{eq:alpha} then give $\eeq=\alpha^{-1}$, so the network has $K$ checkpoints but only a single factor of discrimination. We break detailed balance by multiplying the symmetry-related pair $p_j\to p_{j+1}$ and $\bar p_j\to\bar p_{j+1}$, $1\le j\le K-1$, by a common factor $m>1$, so that the driven cycles carry affinity $\Delta\mu=\ln m$. This turns $\Geq$ into the driven graph $G$, with the same vertices and edges but with $u_j$ replaced by $m\,u_j$ on both driven edges.

\textit{Exact error.}---Two edges are driven, one per wing, but they never interfere in the following sense: the wings meet only at $1$, so the path $T_i$ from any vertex $i$ to $1$ stays inside the wing containing $i$ and can meet only that wing's driven edge---along the drive, against it, or not at all. This shifts $S(T_i)$ from
its undriven value $S_{\mathrm{eq}}(i,1)$ by $+\ln m$, by $-\ln m$, or not at all, and partitions $\Th_1(G)$ into $\Th_-(i)$, $\Th_+(i)$, and $\Th_0(i)$, respectively---the subscript recording the sign of the change in the factor
$e^{-S(T_i)}$, not in the action---so Eq.~\eqref{eq:arboreal} collapses
(Appendix~\ref{app:edge}),
\begin{equation}
  \pi_i\;\propto\;e^{-S_{\mathrm{eq}}(i,1)}
  \Bigl[A_0(i)+m\,A_+(i)+m^{-1}A_-(i)\Bigr],
  \label{eq:three}
\end{equation}
whose \emph{arboreal coefficients} $A_\bullet(i)=w(\Th_\bullet(i))/w(\Th_1(G))$, $\bullet\in\{0,+,-\}$, are the arboreal probabilities of the three classes and sum to one. In the butterfly graph, Fig.~\ref{fig:law}(a), a simple path from $p_K$ to $1$ cannot traverse $p_j\to p_{j+1}$ without revisiting $p_{j+1}$, so $A_-(e)=A_-(\bar e)=0$. Using $A_0=1-A_+$ at each exit and canceling the common normalization gives the error exactly:
\begin{equation}
  \boxed{\;
  \epsilon= \frac{\pi_{\bar e}}{\pi_e}=\eeq\,
  \frac{1+(m-1)A_+(\bar e)}{1+(m-1)A_+(e)}\;,}
  \label{eq:exacterror}
\end{equation}
where $\eeq$ is $e^{-S_{\mathrm{eq}}(\bar e,1)+S_{\mathrm{eq}}(e,1)}$ by
Eq.~\eqref{eq:boltzmann}.

\textit{Reduction to one wing.}---For a vertex set $U$, the \emph{induced} subgraph $G[U]$ is $U$ together with every edge of $G$ joining two of its vertices. Write $C_r=G[\{1,p_1,\ldots,p_r\}]$ for the root and the first $r$
proximal states of the correct wing, and let
\begin{equation}
  Z_r(\theta;m)=\sum_{T\in\Th_1(C_r)} w_{\theta}(T)
  \label{eq:Zdef}
\end{equation}
be the total weight of its rooted spanning trees, $w_{\theta}$ denoting the tree weight with every fall-off rate multiplied by an auxiliary knob $\theta>0$. Therefore, by Eq.~\eqref{eq:alpha}, $\theta=1$ gives the correct wing and $\theta=\alpha$ the incorrect one.

\textit{Two knobs, two arguments.}---The parameter $\theta$ scales the fall-off rates, while $m$ replaces $u_j$ by $m\,u_j$. We therefore write $X(\theta;m)$ with the convention $X(\theta)\equiv X(\theta;1)$. Thus a missing second argument denotes the undriven value. Because the wings meet only at $1$, a rooted spanning tree of $G$ is exactly a pair of rooted trees, one per wing, and weights multiply:
\begin{equation}
  w\bigl(\Th_1(G)\bigr)=Z_K(1;m)\,Z_K(\alpha;m).
  \label{eq:fullfactor}
\end{equation}
Let $B_{K,j}(\theta)$ be the weight of those trees in $\Th_1(C_K)$ whose path from $p_K$ to $1$ contains the reverse driven edge $p_{j+1}\to p_j$. Such a tree cannot also contain $p_j\to p_{j+1}$, which would close a two-cycle, so $B_{K,j}$ carries no $u_j$ and takes no drive argument. Because $T_e$ never leaves the correct wing, $\Th_+(e)$ constrains that wing to be such a tree and leaves the other free, and $\Th_+(\bar e)$ does the reverse, so the same factorization gives
\begin{equation}
  \begin{aligned}
    w\bigl(\Th_+(e)\bigr)&=B_{K,j}(1)\,Z_K(\alpha;m),\\
    w\bigl(\Th_+(\bar e)\bigr)&=B_{K,j}(\alpha)\,Z_K(1;m).
  \end{aligned}
  \label{eq:eventsmaintext}
\end{equation}
Dividing Eq.~\eqref{eq:eventsmaintext} by Eq.~\eqref{eq:fullfactor}, the free wing cancels:
\begin{equation}
  \begin{split}
    A_+(e)&=\frac{w\bigl(\Th_+(e)\bigr)}{w\bigl(\Th_1(G)\bigr)}
           =\frac{B_{K,j}(1)\,Z_K(\alpha;m)}{Z_K(1;m)\,Z_K(\alpha;m)}\\[2pt]
          &=\frac{B_{K,j}(1)}{Z_K(1;m)}\;\equiv\;\Pkj(1;m),
  \end{split}
  \label{eq:Aratio}
\end{equation}
where $\Pkj(\theta;m)\equiv B_{K,j}(\theta)/Z_K(\theta;m)$ and exchanging the wings gives $A_+(\bar e)=\Pkj(\alpha;m)$. The problem is now one-dimensional: the whole drive dependence of Eq.~\eqref{eq:exacterror} sits in the two numbers $\Pkj(1;m)$ and $\Pkj(\alpha;m)$, and a bound on their ratio will propagate to a bound on the error.

\textit{Results.}---We can now state the position law.
\begin{theorem}[Position law]
\label{thm:main}
Fix $K\ge2$, $\alpha>1$ and $1\le j\le K-1$. For every admissible choice of rate constants and every finite $m>1$,
\begin{equation}
  \epsilon\;>\;\alpha^{-(K-j+1)}=\eeq^{\,K-j+1},
  \label{eq:bound}
\end{equation}
and, allowing the rates and $m$ to vary, $\inf\epsilon=\eeq^{\,K-j+1}$.
\end{theorem}

Ignoring orientations, $C_j^{\downarrow}$ has $K-j+2$ vertices and $2(K-j)+1$ edges, so $b_1(C_j^{\downarrow})=K-j$, which is Eq.~\eqref{eq:depth}.

\textit{Proof.}---Three steps are required. \textit{(i) The forced tail} [Fig.~\ref{fig:mech}(a)].
\begin{lemma}[Forced tail]
\label{lem:tail}
For every $\theta>0$,
$\;B_{K,j}(\theta)=\bigl(\prod_{r=j}^{K-1}v_r\bigr)\,Z_j(\theta)$.
\end{lemma}
\begin{proof}
Consider a tree counted by $B_{K,j}$. For each $r>j$, once its path from $p_K$ to $1$ reaches $p_r$, it cannot move directly to $1$, since it has not yet reached $p_j$, and, if $r<K$, it cannot return to $p_{r+1}$, since the path is simple. Its next edge must therefore be $p_r\to p_{r-1}$. Iterating from $r=K$ down to $r=j+1$ forces the tail
\[
  p_K\to p_{K-1}\to\cdots\to p_j,
\]
whose weight is $\prod_{r=j}^{K-1}v_r$ [Fig.~\ref{fig:mech}(a)]. Its last edge $p_{j+1}\to p_j$ excludes $p_j\to p_{j+1}$, so deleting $p_{j+1},\ldots,p_K$ and their outgoing edges leaves a tree in $\Th_1(C_j)$. Conversely, adjoining the forced tail to any tree in $\Th_1(C_j)$ produces a tree counted by $B_{K,j}$. These operations are inverse, and adjoining the tail multiplies each tree weight by $\prod_{r=j}^{K-1}v_r$.
\end{proof}

Both sides of the identity in Lemma~\ref{lem:tail} are free of the drive and dividing that identity by $Z_K(\theta;m)$ we arrive at
\begin{equation}
  \Pkj(\theta;m)=\Bigl(\prod\nolimits_{r=j}^{K-1}v_r\Bigr)
  \frac{Z_j(\theta)}{Z_K(\theta;m)}.
  \label{eq:Pfactor}
\end{equation}

\textit{(ii) One layer at a time.} Let $F_r$ be the total weight of the spanning forests of $C_r$ rooted at $\{1,p_r\}$---acyclic, with $p_r$, like $1$, having no outgoing edge [Fig.~\ref{fig:objects}(c)]---and $g_r=F_r/Z_r$. Asking what the last vertex $p_r$ points to---to $1$, or back to $p_{r-1}$---and, in the first case, whether $p_{r-1}$ points at $p_r$, gives two elementary counting identities, depicted in Fig.~\ref{fig:mech}(b,c) and
derived in Appendix~\ref{app:recursions},
\begin{equation}
  \begin{aligned}
    Z_r&=\theta d_r\bigl(Z_{r-1}+u_{r-1}F_{r-1}\bigr)+v_{r-1}Z_{r-1},\\
    F_r&=Z_{r-1}+u_{r-1}F_{r-1},
  \end{aligned}
  \label{eq:recursions}
\end{equation}
for $r\ge2$, with $Z_1=\theta d_1$ and $F_1=1$. Both arguments are suppressed here since the recursions hold verbatim for either wing, and in $G$ as in $\Geq$, with $u_j$ replaced by $m\,u_j$ in the former. Dividing the second identity by the first turns Eq.~\eqref{eq:recursions} into a scalar continued fraction,
\begin{equation}
  g_r=\cfrac{1}{\theta d_r+\cfrac{v_{r-1}}
  {1+u_{r-1}g_{r-1}}},
  \qquad g_1=\frac{1}{\theta d_1}.
  \label{eq:cf}
\end{equation}

\begin{lemma}[One-layer monotonicity]
\label{lem:mono}
For every strictly positive choice of rate constants---hence, at each fixed $m$, for the driven wing exactly as for the undriven one---$g_r(\theta;m)$ is strictly decreasing on $\theta>0$ for every $r$. Consequently $q_r(\theta;m)\equiv\theta\,Z_{r-1}(\theta;m)/Z_r(\theta;m)$ is strictly increasing on $\theta>0$ for every $r\ge2$.
\end{lemma}

\begin{proof}
Induct on Eq.~\eqref{eq:cf}. For $r=1$, $g_1=1/(\theta d_1)$ is strictly decreasing. If $g_{r-1}$ is strictly decreasing, then $1+u_{r-1}g_{r-1}$ decreases, so $v_{r-1}/(1+u_{r-1}g_{r-1})$ increases, while $\theta d_r$ strictly increases; the denominator of Eq.~\eqref{eq:cf} therefore strictly increases and $g_r$ strictly decreases. For $q_r$, the first identity in Eq.~\eqref{eq:recursions} reads $Z_r=Z_{r-1}\left[v_{r-1}+\theta d_r(1+u_{r-1}g_{r-1})\right]$, whence $q_r=\left[v_{r-1}/\theta+d_r+d_r u_{r-1}g_{r-1}\right]^{-1}$, in which both $\theta$-dependent terms strictly decrease.
\end{proof}
\begin{figure}[t]
  \centering
  \includegraphics[width=\linewidth]{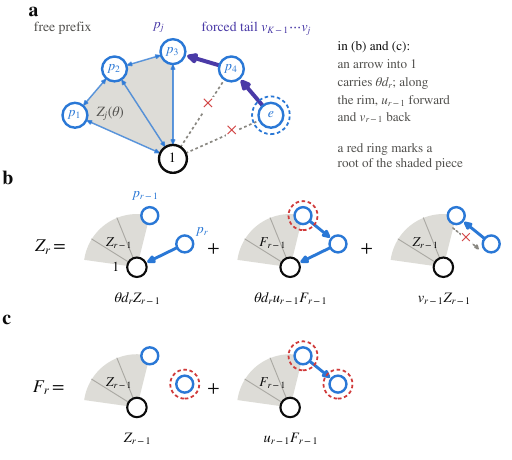}
\caption{\textbf{Counting trees one vertex at a time.} Every nonroot vertex of a spanning tree spends exactly one outgoing edge. \textbf{(a)}~The forced tail (Lemma~\ref{lem:tail}) for $K=5$, $j=3$: one wing of Fig.~\ref{fig:law}(a), drawn upright. Falling off at $e$ or at $p_4$ (red crosses) would end the path from $e$ to $1$ before it reached $p_3$, so the tail $p_5\to p_4\to p_3$ is forced, of weight $v_4v_3$, while the trees on $\{1,p_1,p_2,p_3\}$ are left free: $B_{5,3}=v_4v_3\,Z_3$. \textbf{(b)}~$Z_r$, the trees of $C_r$ rooted at $1$. Its last vertex $p_r$
  spends its edge either on $1$ (weight $\theta d_r$) or on $p_{r-1}$ (weight $v_{r-1}$). In the first case ask whether $p_{r-1}$ points back at $p_r$: it does not ($Z_{r-1}$), or it does ($u_{r-1}F_{r-1}$). In the second case, the reverse edge is barred (dotted), because a tree has no two-cycle, and $Z_{r-1}$ remains. These three cases are the three terms of the first line of Eq.~\eqref{eq:recursions}, in that order. \textbf{(c)}~$F_r$, the forests of $C_r$ rooted at $\{1,p_r\}$. Now $p_r$ spends nothing, so only $p_{r-1}$ is in question: it does not point at $p_r$ ($Z_{r-1}$), or it does ($u_{r-1}F_{r-1}$)---the two terms of the second line.}
  \label{fig:mech}
\end{figure}

\textit{(iii) Assembly.} Set $\nu=K-j$. Telescoping $q_r=\theta Z_{r-1}/Z_r$ over $r=j+1,\ldots,K$ gives
$Z_j/Z_K=\theta^{-\nu}\prod_{r>j}q_r$, and inserting this in Eq.~\eqref{eq:Pfactor} yields the central identity
\begin{equation}
  \frac{\Pkj(\alpha;m)}{\Pkj(1;m)}
  \;=\;\alpha^{-\nu}\prod_{r=j+1}^{K}\frac{q_r(\alpha;m)}{q_r(1;m)}
  \;>\;\alpha^{-\nu},
  \label{eq:centralratio}
\end{equation}
every factor exceeding one by Lemma~\ref{lem:mono}---and doing so at every value of $m$, which is what makes the bound uniform in the drive.

Finally, put $a=\Pkj(\alpha;m)$, $b=\Pkj(1;m)$, $c=m-1>0$ and $s=\alpha^{-\nu}$. All rates are positive, so $b>0$, while trees with $p_K\to1$ are excluded from $B_{K,j}$ and carry positive weight, so $b<1$; combined with $a>sb$, which is Eq.~\eqref{eq:centralratio}, this forces $(1+ca)/(1+cb)>(1+sc)/(1+c)=s+(1-s)/m$ (Appendix~\ref{app:refined}). Hence Eq.~\eqref{eq:exacterror} sharpens Eq.~\eqref{eq:bound} to
\begin{equation}
  \epsilon\;>\;\eeq^{\,K-j+1}+e^{-\Delta\mu}\bigl(\eeq-\eeq^{\,K-j+1}\bigr)
  \;>\;\eeq^{\,K-j+1},
  \label{eq:refined}
\end{equation}
since $\eeq>\eeq^{\,K-j+1}$. This is Eq.~\eqref{eq:bound}. Appendix~\ref{app:sharp} gives sharpness. $\square$

\textit{The two endpoints.}---The two extreme placements are worth naming.
\begin{corollary}[Global synergy]
\label{cor:synergy}
For $j=1$, $b_1(C_1^{\downarrow})=K-1$ and $\inf\epsilon=\eeq^{K}$: all $K-1$ independent cycles cooperate and a single driven transition unlocks the full $K$-fold discriminatory capacity of the network.
\end{corollary}

\begin{corollary}[Information-theoretic short circuit]
\label{cor:short}
For $j=K-1$, $b_1(C_{K-1}^{\downarrow})=1$ and $\inf\epsilon=\eeq^{2}$, independently of $K$: however many discriminatory checkpoints lie upstream, they are stranded, and an arbitrarily deep cascade performs no better than a two-checkpoint scheme.
\end{corollary}

\begin{figure}[t]
  \centering
  \includegraphics[width=\linewidth]{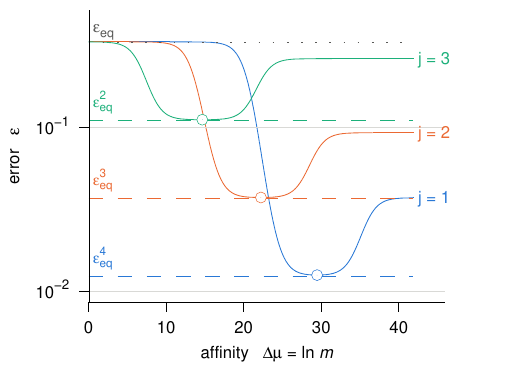}
\caption{\textbf{Position caps what affinity can buy.}
Error against affinity $\Delta\mu=\ln m$ for a $K=4$ network at $\alpha=3$, one curve per position $j$ of the driven transition; rates from the family of Eq.~\eqref{eq:family} at $\delta=10^{-3}$, each curve using the member belonging to its own $j$. Each curve starts at $\eeq$, reaches its best error at $\Delta\mu^{*}=\ln m^{*}$ (circles, Theorem~\ref{thm:opt}), then rises again: past the optimum, more drive makes the error worse. Dashed lines are the floors $\eeq^{\,K-j+1}$ of Theorem~\ref{thm:main}, which no curve crosses at any affinity. Curves are computed from Eq.~\eqref{eq:master} and checked against a direct numerical solution of the nine-state master equation.}
\label{fig:tradeoff}
\end{figure}

\textit{The energy--accuracy relation.}---Theorem~\ref{thm:main} fixes the limit set by \emph{where} the energy is spent. We now ask how the error depends on the driving affinity $\Delta\mu=\ln m$. The whole $m$-dependence of Eq.~\eqref{eq:Pfactor} sits in $Z_K(\theta;m)$. Since a rooted spanning tree spends exactly one outgoing edge per nonroot vertex, it can have $p_j\to p_{j+1}$ at most once, so $Z_K(\theta;m)$ is affine in $m$. Accordingly, $\Pkj(\theta;m)=B_{K,j}(\theta)/Z_K(\theta;m)$ can be expressed as (Appendix~\ref{app:dilution})
\begin{equation}
 \Pkj(\theta;m)=\frac{P_\theta}{1+(m-1)\,R_\theta}\;,
  \qquad
  P_\theta\equiv\Pkj(\theta),
  \label{eq:dilute}
\end{equation}
where \(R_\theta\) is the arboreal probability, in $\Geq$, that a rooted spanning tree of $C_K$ contains \(p_j\to p_{j+1}\). Here the convention earns its keep: the left-hand side is the driven quantity of Eq.~\eqref{eq:Aratio}, while $P_\theta$ and $R_\theta$ on the right belong to $\Geq$.

\textit{The drive dependence, exactly.}---Using Eq.~\eqref{eq:Aratio} and substituting Eq.~\eqref{eq:dilute} into Eq.~\eqref{eq:exacterror}, we find that the entire affinity dependence of a $K$-checkpoint network collapses onto four equilibrium numbers,
\begin{equation}
  \epsilon(m)=\eeq\,\frac{\Phi_\alpha(c)}{\Phi_1(c)},
  \qquad
  \Phi_\theta(c)=\frac{1+c\,(R_\theta+P_\theta)}{1+c\,R_\theta},
  \label{eq:master}
\end{equation}
with $c=m-1$. Each $\Phi_\theta$ increases monotonically from $1$ at $c=0$ to the saturation gain $\Gamma_\theta\equiv\Phi_\theta(\infty)=1+P_\theta/R_\theta$ as $c\to\infty$. All four numbers---$P_1,P_\alpha,R_1,R_\alpha$---belong to $\Geq$: the drive appears in Eq.~\eqref{eq:master} through $c$ and nowhere else.

\begin{lemma}[Initial descent, and no antiproofreading]
\label{lem:first}
For every admissible choice of rates, $P_\alpha<P_1$, so that $\left.\mathrm{d}\epsilon/\mathrm{d}m\right|_{m=1}=\eeq(P_\alpha-P_1)<0$. Moreover $\Pkj(\alpha;m)<\Pkj(1;m)$ at every $m$, so $\epsilon(m)<\eeq$ for every $m>1$. In other words, antiproofreading ($\epsilon(m)>\eeq$) is impossible (Appendix~\ref{app:signs}).
\end{lemma}

By Lemma~\ref{lem:tail}, $P_\theta$ factorizes over layers as $\prod_{r=j+1}^{K}\rho_r(\theta)$, with $\rho_r=v_{r-1}Z_{r-1}/Z_r$. In Appendix~\ref{app:signs}, we show that $\alpha^{-\nu}<P_\alpha/P_1<1$. Lemma~\ref{lem:first} says that the error falls as soon as the system departs from equilibrium. Whether it then rises again past an optimum or falls monotonically to its limiting value is settled by a single sign [Fig.~\ref{fig:tradeoff}]:

\begin{theorem}[Optimal drive]
\label{thm:opt}
Fix $K\ge2$, $\alpha>1$, $1\le j\le K-1$ and any admissible rates, and put, in
$\Geq$,
\begin{equation}
  \begin{aligned}
    Q_0&=P_\alpha-P_1, \qquad Q_1=2\bigl(P_\alpha R_1-P_1R_\alpha\bigr),\\
    Q_2&=P_\alpha R_1\bigl(R_1+P_1\bigr)-P_1R_\alpha\bigl(R_\alpha+P_\alpha\bigr).
  \end{aligned}
  \label{eq:Qcoef}
\end{equation}
Then $\mathrm{d}\epsilon/\mathrm{d}m$ has the sign of $Q(c)=Q_2c^2+Q_1c+Q_0$. Moreover, $Q_0<0$ (Lemma~\ref{lem:first}), $Q_1\le0$ (Appendix~\ref{app:Q1}). If $Q_2>0$, then $\epsilon$ strictly decreases on
$(1,m^*)$ and strictly increases on $(m^*,\infty)$, where $m^*=1+c^*$ and
\begin{equation}
  c^*=\frac{-Q_1+\sqrt{Q_1^{2}-4Q_0Q_2}}{2Q_2}
  \label{eq:cstar}
\end{equation}
is the unique positive root of $Q$; the minimum is global and equals
\begin{equation}
  \epsilon(m^*)=\eeq\,\frac{P_\alpha}{P_1}
  \left(\frac{1+c^*R_1}{1+c^*R_\alpha}\right)^{2}.
  \label{eq:valstar}
\end{equation}
(See, Fig.~\ref{fig:tradeoff}). If $Q_2\le0$, then $\epsilon$ is strictly decreasing on $(1,\infty)$, toward $\eeq\,\Gamma_\alpha/\Gamma_1$ (Appendix~\ref{app:optimum}).
\end{theorem}

\textit{Conclusions.}---We have shown that energy expenditure acts as a topological gate: the placement of the drive determines which part of a network's cycle space can contribute to information processing. In the butterfly graph, the effective proofreading depth is one plus the first Betti number of the subgraph lying at or downstream of the energetic edge. Moreover, past an optimal affinity, further driving can impair rather than improve error correction. What accuracy a nonequilibrium network can reach is therefore governed not only by how far it is driven from equilibrium, but also by where the drive enters it.


\bibliographystyle{apsrev4-2}
\bibliography{refs}

@incollection{Pauling1957,
  author    = {Linus Pauling},
  title     = {The probability of errors in the process of synthesis of protein molecules},
  booktitle = {Festschrift Arthur Stoll},
  pages     = {597--602},
  publisher = {Birkh\"auser Verlag},
  address   = {Basel},
  year      = {1957}
}

@article{Loftfield1972,
  author  = {R. B. Loftfield and D. Vanderjagt},
  title   = {The frequency of errors in protein biosynthesis},
  journal = {Biochem. J.},
  volume  = {128},
  pages   = {1353--1356},
  year    = {1972},
  doi     = {10.1042/bj1281353}
}

@article{Hopfield1974,
  author  = {J. J. Hopfield},
  title   = {Kinetic proofreading: A new mechanism for reducing errors in biosynthetic processes requiring high specificity},
  journal = {Proc. Natl. Acad. Sci. U.S.A.},
  volume  = {71},
  pages   = {4135--4139},
  year    = {1974},
  doi     = {10.1073/pnas.71.10.4135}
}

@article{Ninio1975,
  author  = {J. Ninio},
  title   = {Kinetic amplification of enzyme discrimination},
  journal = {Biochimie},
  volume  = {57},
  pages   = {587--595},
  year    = {1975},
  doi     = {10.1016/S0300-9084(75)80139-8}
}

@article{Bennett1979,
  author  = {C. H. Bennett},
  title   = {Dissipation-error tradeoff in proofreading},
  journal = {BioSystems},
  volume  = {11},
  pages   = {85--91},
  year    = {1979},
  doi     = {10.1016/0303-2647(79)90003-0}
}

@article{Ehrenberg1980,
  author  = {M. Ehrenberg and C. Blomberg},
  title   = {Thermodynamic constraints on kinetic proofreading in biosynthetic pathways},
  journal = {Biophys. J.},
  volume  = {31},
  pages   = {333--358},
  year    = {1980},
  doi     = {10.1016/S0006-3495(80)85063-6}
}

@article{McKeithan1995,
  author  = {T. W. McKeithan},
  title   = {Kinetic proofreading in {T}-cell receptor signal transduction},
  journal = {Proc. Natl. Acad. Sci. U.S.A.},
  volume  = {92},
  pages   = {5042--5046},
  year    = {1995},
  doi     = {10.1073/pnas.92.11.5042}
}

@article{Murugan2012,
  author  = {A. Murugan and D. A. Huse and S. Leibler},
  title   = {Speed, dissipation, and error in kinetic proofreading},
  journal = {Proc. Natl. Acad. Sci. U.S.A.},
  volume  = {109},
  pages   = {12034--12039},
  year    = {2012},
  doi     = {10.1073/pnas.1119911109}
}

@article{Murugan2014,
  author  = {A. Murugan and D. A. Huse and S. Leibler},
  title   = {Discriminatory proofreading regimes in nonequilibrium systems},
  journal = {Phys. Rev. X},
  volume  = {4},
  pages   = {021016},
  year    = {2014},
  doi     = {10.1103/PhysRevX.4.021016}
}

@article{Sartori2013,
  author  = {P. Sartori and S. Pigolotti},
  title   = {Kinetic versus energetic discrimination in biological copying},
  journal = {Phys. Rev. Lett.},
  volume  = {110},
  pages   = {188101},
  year    = {2013},
  doi     = {10.1103/PhysRevLett.110.188101}
}

@article{Sartori2015,
  author  = {P. Sartori and S. Pigolotti},
  title   = {Thermodynamics of error correction},
  journal = {Phys. Rev. X},
  volume  = {5},
  pages   = {041039},
  year    = {2015},
  doi     = {10.1103/PhysRevX.5.041039}
}

@article{Rao2015,
  author  = {R. Rao and L. Peliti},
  title   = {Thermodynamics of accuracy in kinetic proofreading: dissipation and efficiency trade-offs},
  journal = {J. Stat. Mech.},
  volume  = {2015},
  pages   = {P06001},
  year    = {2015},
  doi     = {10.1088/1742-5468/2015/06/P06001}
}

@article{Banerjee2017,
  author  = {K. Banerjee and A. B. Kolomeisky and O. A. Igoshin},
  title   = {Elucidating interplay of speed and accuracy in biological error correction},
  journal = {Proc. Natl. Acad. Sci. U.S.A.},
  volume  = {114},
  pages   = {5183--5188},
  year    = {2017},
  doi     = {10.1073/pnas.1614838114}
}

@article{Lan2012,
  author  = {G. Lan and P. Sartori and S. Neumann and V. Sourjik and Y. Tu},
  title   = {The energy--speed--accuracy trade-off in sensory adaptation},
  journal = {Nat. Phys.},
  volume  = {8},
  pages   = {422--428},
  year    = {2012},
  doi     = {10.1038/nphys2276}
}

@article{Cetiner2022,
  author  = {U. \c{C}etiner and J. Gunawardena},
  title   = {Reformulating nonequilibrium steady states and generalized {H}opfield discrimination},
  journal = {Phys. Rev. E},
  volume  = {106},
  pages   = {064128},
  year    = {2022},
  doi     = {10.1103/PhysRevE.106.064128}
}

@article{Gunawardena2012,
  author  = {J. Gunawardena},
  title   = {A linear framework for time-scale separation in nonlinear biochemical systems},
  journal = {PLoS ONE},
  volume  = {7},
  pages   = {e36321},
  year    = {2012},
  doi     = {10.1371/journal.pone.0036321}
}

@article{Mirzaev2013,
  author  = {I. Mirzaev and J. Gunawardena},
  title   = {Laplacian dynamics on general graphs},
  journal = {Bull. Math. Biol.},
  volume  = {75},
  pages   = {2118--2149},
  year    = {2013},
  doi     = {10.1007/s11538-013-9884-8}
}

@article{Schnakenberg1976,
  author  = {J. Schnakenberg},
  title   = {Network theory of microscopic and macroscopic behavior of master equation systems},
  journal = {Rev. Mod. Phys.},
  volume  = {48},
  pages   = {571--585},
  year    = {1976},
  doi     = {10.1103/RevModPhys.48.571}
}

@article{Lebowitz1999,
  author  = {J. L. Lebowitz and H. Spohn},
  title   = {A {G}allavotti--{C}ohen-type symmetry in the large deviation functional for stochastic dynamics},
  journal = {J. Stat. Phys.},
  volume  = {95},
  pages   = {333--365},
  year    = {1999},
  doi     = {10.1023/A:1004589714161}
}

@article{Seifert2008,
  author  = {U. Seifert},
  title   = {Stochastic thermodynamics: principles and perspectives},
  journal = {Eur. Phys. J. B},
  volume  = {64},
  pages   = {423--431},
  year    = {2008},
  doi     = {10.1140/epjb/e2008-00001-9}
}

@article{Wong2018,
  author  = {F. Wong and A. Amir and J. Gunawardena},
  title   = {Energy-speed-accuracy relation in complex networks for biological discrimination},
  journal = {Phys. Rev. E},
  volume  = {98},
  pages   = {012420},
  year    = {2018},
  doi     = {10.1103/PhysRevE.98.012420}
}

@article{Hopfield1976,
  author  = {J. J. Hopfield and T. Yamane and V. Yue and S. M. Coutts},
  title   = {Direct experimental evidence for kinetic proofreading in amino acylation of {tRNA}$^{\mathrm{Ile}}$},
  journal = {Proc. Natl. Acad. Sci. U.S.A.},
  volume  = {73},
  pages   = {1164--1168},
  year    = {1976},
  doi     = {10.1073/pnas.73.4.1164}
}

@article{cetiner2023universal,
  title   = {Universal bounds on information-processing capabilities of markov processes},
  author  = {Cetiner, Ugur and Gunawardena, Jeremy},
  journal = {arXiv preprint arXiv:2310.10584},
  year    = {2023}
}

\clearpage
\Appendix
\app{app:edge}{A single energetic edge and the arboreal coefficients}%
Perturb an equilibrium process by replacing one rate $\ell(z_1\to z_2)$ by $m\,\ell(z_1\to z_2)$, $m>1$. We call $z_1\to z_2$ the \emph{energetic edge}. By Eq.~\eqref{eq:action}, the action of a path $P$ from $i$ to $k$ changes by $\ln m$ if $P$ traverses $z_1\to z_2$, by $-\ln m$ if it traverses $z_2\to z_1$, and not at all otherwise:
\begin{equation}
  S(P)=S_{\mathrm{eq}}(i,k)+
  \begin{cases}
    \ln m, & z_1\to z_2\in P,\\
    -\ln m, & z_2\to z_1\in P,\\
    0, & \text{otherwise.}
  \end{cases}
  \label{eq:threevalues}
\end{equation}
No matter how large the graph, only three values occur. Accordingly, for each vertex $i$, partition the rooted spanning trees by how $T_i$ meets the energetic edge,
\begin{equation}
  \begin{aligned}
    \Th_+(i)&=\{T\in\Th_1(G)\mid z_2\to z_1\in T_i\},\\
    \Th_-(i)&=\{T\in\Th_1(G)\mid z_1\to z_2\in T_i\},\\
    \Th_0(i)&=\Th_1(G)\setminus\bigl(\Th_+(i)\sqcup\Th_-(i)\bigr),
  \end{aligned}
  \label{eq:partition}
\end{equation}
and define the arboreal coefficients as the corresponding probabilities,
\begin{equation}
  A_\bullet(i)=\frac{w\bigl(\Th_\bullet(i)\bigr)}{w\bigl(\Th_1(G)\bigr)},
  \qquad \bullet\in\{0,+,-\},
  \label{eq:acoeffs}
\end{equation}
so that $A_0(i)+A_+(i)+A_-(i)=1$ for every $i$. A tree in $\Th_+(i)$ carries $e^{-S(T_i)}=m\,e^{-S_{\mathrm{eq}}(i,1)}$, one in $\Th_-(i)$ carries $m^{-1}e^{-S_{\mathrm{eq}}(i,1)}$, and one in $\Th_0(i)$ carries $e^{-S_{\mathrm{eq}}(i,1)}$; grouping the sum in Eq.~\eqref{eq:arboreal} into these three classes gives Eq.~\eqref{eq:three}. The arboreal coefficients thus package the entire nonequilibrium problem exactly into three numbers per state.

For the butterfly there is one energetic edge per wing, but since the wings meet only at the root, the path from either exit to $1$ stays inside its own wing. The decomposition therefore applies to $e$ and $\bar e$ separately. Moreover, a path from $p_K$ to $1$ cannot contain $p_j\to p_{j+1}$: to reach $p_j$ from downstream it must already have passed through $p_{j+1}$, and a path does not repeat a vertex. Hence $A_-(e)=A_-(\bar e)=0$, and with $A_0=1-A_+$ Eq.~\eqref{eq:three} gives $\pi_e\propto e^{-S_{\mathrm{eq}}(e,1)}[1+(m-1)A_+(e)]$ and likewise for $\bar e$. Their ratio is Eq.~\eqref{eq:exacterror}, the prefactor $e^{-S_{\mathrm{eq}}(\bar e,1)+S_{\mathrm{eq}}(e,1)}$ being exactly $\eeq$ by Eq.~\eqref{eq:alpha} and the equality of the binding rates.

\app{app:recursions}{The two counting recursions}%
\textit{The recursion for $Z_r$.} Let $T$ be a spanning tree of $C_r$ rooted at $1$. Only two edges leave $p_r$, so there are two cases and the first splits again. \emph{Case A: $p_r\to1$}, of weight $\theta d_r$. Because $p_r$ now reaches the root directly, $p_{r-1}$ is free to point at it without creating a cycle. Ask whether it does. \emph{A1: it does not.} Then no edge enters $p_r$, so $p_1,\ldots,p_{r-1}$ must reach $1$ inside $C_{r-1}$: what remains is a spanning tree of $C_{r-1}$ rooted at $1$, and summing over all of them contributes $Z_{r-1}$. \emph{A2: it does}, at weight $u_{r-1}$. Then $p_{r-1}$ has spent its outgoing edge, and $p_1,\ldots,p_{r-2}$ must each reach either $1$ or $p_{r-1}$. What remains is a spanning forest of $C_{r-1}$ rooted at $\{1,p_{r-1}\}$, and summing gives $F_{r-1}$.

\emph{Case B: $p_r\to p_{r-1}$}, of weight $v_{r-1}$. Now $p_{r-1}\to p_r$ is barred, since it would close the two-cycle $p_{r-1}\to p_r\to p_{r-1}$. Again no edge enters $p_r$, every vertex reaches $1$ inside $C_{r-1}$, and summing gives $Z_{r-1}$. Collecting the three contributions, $Z_r=\theta d_r\bigl(Z_{r-1}+u_{r-1}F_{r-1}\bigr)+v_{r-1}Z_{r-1}$, the first identity of Eq.~\eqref{eq:recursions}.

\smallskip\noindent
\textit{The recursion for $F_r$.} Now take a spanning forest of $C_r$ rooted at $\{1,p_r\}$. Here $p_r$ spends no outgoing edge, so ask instead about $p_{r-1}$, the only vertex that can point at it. If $p_{r-1}\to p_r$, of weight $u_{r-1}$, the rest is a spanning forest of $C_{r-1}$ rooted at $\{1,p_{r-1}\}$, contributing $F_{r-1}$. If not, nothing reaches $p_r$ and each of $p_1,\ldots,p_{r-1}$ must reach $1$ inside $C_{r-1}$, contributing $Z_{r-1}$. That is the second identity of Eq.~\eqref{eq:recursions}.

\smallskip\noindent
\textit{Dividing through.} Set $\beta_{r-1}=1+u_{r-1}g_{r-1}$, so that $Z_{r-1}+u_{r-1}F_{r-1}=Z_{r-1}\beta_{r-1}$ and Eq.~\eqref{eq:recursions} becomes the compact pair
\begin{equation}
  Z_r=Z_{r-1}\bigl[v_{r-1}+\theta d_r\beta_{r-1}\bigr],
  \qquad
  F_r=Z_{r-1}\beta_{r-1}.
  \label{eq:pairEM}
\end{equation}
Dividing the second by the first gives $g_r=\beta_{r-1}/(v_{r-1}+\theta d_r\beta_{r-1})$; dividing numerator and denominator by $\beta_{r-1}$ and substituting $\beta_{r-1}=1+u_{r-1}g_{r-1}$ gives the continued fraction Eq.~\eqref{eq:cf}. Two consequences used below are $g_r\le1/(\theta d_r)$, because the second term in the denominator of Eq.~\eqref{eq:cf} is positive, and the per-layer ratio
\begin{equation}
  \rho_r(\theta)\;\equiv\;\frac{v_{r-1}Z_{r-1}(\theta)}{Z_r(\theta)}
  \;=\;\frac{v_{r-1}}{\theta}\,q_r(\theta),
  \qquad r\ge2,
  \label{eq:cr}
\end{equation}
whose product over $r=j+1,\ldots,K$ is $P_\theta$ by Lemma~\ref{lem:tail}.

\begin{figure*}[t]
  \centering
  \includegraphics[width=\linewidth]{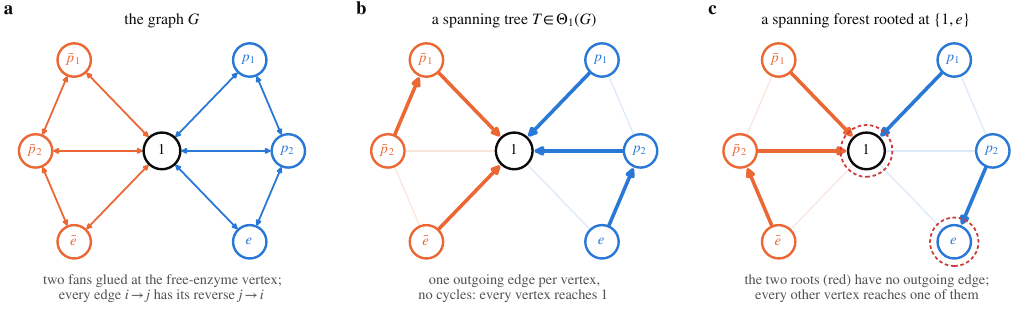}
  \caption{\textbf{The objects the proof counts}, for the core butterfly graph with $K=3$. \textbf{(a)}~The graph itself: states are vertices, transitions are labeled directed edges, and every edge is accompanied by its reverse. \textbf{(b)}~One of its spanning trees rooted at the free-enzyme vertex $1$: every vertex except the root has exactly one outgoing edge, and the absence of cycles means that every vertex has a unique directed path to $1$. The arboreal distribution, Eq.~\eqref{eq:arboreal}, assigns each such tree the probability $w(T)/w(\Th_1(G))$. \textbf{(c)}~One of its spanning forests rooted at $\{1,e\}$ (roots circled in red): the two roots have no outgoing edges, and every other vertex reaches one of them. The sums of the weights over the one-wing analogs of these two families are $Z_r$ and $F_r$, respectively.}
  \label{fig:objects}
\end{figure*}

\app{app:refined}{The refined bound}%
Write $a=\Pkj(\alpha;m)$, $b=\Pkj(1;m)$, $c=m-1>0$ and $s=\alpha^{-\nu}\in(0,1)$. All rates are positive, so $b>0$, while a rooted spanning tree of $C_K$ with $p_K\to1$ is not counted by $B_{K,j}$ and carries positive weight, so $b<1$; and $a>sb$ is Eq.~\eqref{eq:centralratio}. Then $b-a<b(1-s)<1-s$, which makes the first two terms in the bracket below sum to a positive number, while the third is $c(a-sb)>0$:
\begin{equation}
  \frac{1+ca}{1+cb}-\frac{1+sc}{1+c}
  =\frac{c\bigl[(1-s)-(b-a)+c(a-sb)\bigr]}{(1+cb)(1+c)}>0.
  \label{eq:elementary}
\end{equation}
Because $(1+sc)/(1+c)=s+(1-s)/m$ and $m=e^{\Delta\mu}$, inserting
Eq.~\eqref{eq:elementary} into Eq.~\eqref{eq:exacterror} gives Eq.~\eqref{eq:refined}.

\app{app:dilution}{Affine dependence, dilution, and the master formula}%
Recall the convention: a second argument displays the drive, its absence means $\Geq$. The two graphs have the same vertices and the same directed edges, and hence the same set $\Th_1(C_K)$ of rooted spanning trees; they differ only in the label of $p_j\to p_{j+1}$, which is $u_j$ in $\Geq$ and $m\,u_j$ in $G$. Everything below is that one substitution, tracked.

Let $\mathcal E\subset\Th_1(C_K)$ be the trees that use the driven edge, $\mathcal E=\{T:(p_j\to p_{j+1})\in T\}$. A tree assigns exactly one outgoing edge to each nonroot vertex, so $p_j\to p_{j+1}$ occurs in $T$ once or not at all, and its weight is either linear in $u_j$ or independent of it. Replacing $u_j$ by $m\,u_j$ therefore multiplies $w_\theta(T)$ by $m$ on $\mathcal E$ and leaves it alone off $\mathcal E$:
\begin{equation}
  \begin{aligned}
    Z_K(\theta;m)&=Z_K(\theta)+(m-1)\,w_\theta(\mathcal E)\\
    &=Z_K(\theta)\bigl[1+(m-1)R_\theta\bigr],
  \end{aligned}
  \label{eq:affine}
\end{equation}
every weight on the right-hand side being one of $\Geq$, with
\begin{equation}
  R_\theta=\frac{w_\theta(\mathcal E)}{Z_K(\theta)},
  \qquad
  P_\theta=\Pkj(\theta)=\frac{B_{K,j}(\theta)}{Z_K(\theta)}.
  \label{eq:PRdef}
\end{equation}
Dividing $B_{K,j}(\theta)$ by Eq.~\eqref{eq:affine} gives the dilution relation Eq.~\eqref{eq:dilute}. Both $P_\theta$
and $R_\theta$ lie in $(0,1)$, and their events are disjoint, so $P_\theta+R_\theta<1$: the all-fall-off tree, in which every $p_i$ points at $1$, lies in neither class and has positive weight. Inserting Eq.~\eqref{eq:dilute} into Eq.~\eqref{eq:exacterror} and clearing denominators gives Eq.~\eqref{eq:master}. Nothing here uses the butterfly structure: the affine dependence and the dilution relation hold for a single driven edge in any graph.

\app{app:signs}{Monotonicity in $\theta$ and Lemma~\ref{lem:first}}%
This appendix and the next rest on one function: set $h_r(\theta)=\theta\,g_r(\theta)$, which obeys a recursion of the same form; multiplying Eq.~\eqref{eq:cf} by $\theta$ and using $\theta(1+u_{r-1}g_{r-1})=\theta+u_{r-1}h_{r-1}$ gives
\begin{equation}
  h_r=\Bigl[\,d_r+\frac{v_{r-1}}{\theta+u_{r-1}h_{r-1}}\,\Bigr]^{-1},
  \qquad h_1=\frac{1}{d_1}.
  \label{eq:hrec}
\end{equation}
Induct: $h_1$ is constant, and if $h_{r-1}$ is nondecreasing, then $\theta+u_{r-1}h_{r-1}$ strictly increases, so the bracket strictly decreases and $h_r$ strictly increases. Hence $h_r$ is nondecreasing for every $r$, and
constant if and only if $r=1$.

\textit{Proof of Lemma~\ref{lem:first}.} By Eq.~\eqref{eq:pairEM},
$Z_r/Z_{r-1}=v_{r-1}+d_r\bigl(\theta+u_{r-1}h_{r-1}\bigr)$, strictly increasing in $\theta$, so $\rho_r(\theta)$ of Eq.~\eqref{eq:cr} strictly decreases, and $\rho_r(\alpha)<\rho_r(1)$. Lemma~\ref{lem:tail} and the telescoping of $Z_j/Z_K$ give $P_\theta=\prod_{r=j+1}^{K}\rho_r(\theta)$, so $P_\alpha/P_1=\prod_r\rho_r(\alpha)/\rho_r(1)<1$: each factor is the ratio of the two values of a single decreasing function at $\theta=\alpha$ and $\theta=1$. The last expression in Eq.~\eqref{eq:cr}, $\rho_r=(v_{r-1}/\theta)\,q_r$, bounds those factors from below: it gives $\rho_r(\alpha)/\rho_r(1)=\alpha^{-1}q_r(\alpha)/q_r(1)>\alpha^{-1}$, because $q_r$ increases by Lemma~\ref{lem:mono}. Hence $\alpha^{-\nu}<P_\alpha/P_1<1$, and $\mathrm{d}\epsilon/\mathrm{d}m|_{m=1}=\eeq(P_\alpha-P_1)$ from Eq.~\eqref{eq:master} at $c=0$. Finally, since $c=m-1>0$, Eq.~\eqref{eq:exacterror} gives $\epsilon<\eeq$ exactly when $\Pkj(\alpha;m)<\Pkj(1;m)$. Nothing above used $u_j$ beyond its positivity, so every step holds with $u_j$ replaced by $m\,u_j$: $\rho_r(\alpha;m)<\rho_r(1;m)$ for each $r>j$, and multiplying these gives $\Pkj(\alpha;m)<\Pkj(1;m)$.

\app{app:Q1}{Sign of $Q_1$}%
Cut the wing between $p_j$ and $p_{j+1}$. Downstream of the cut lies $C_{j+1}^{\downarrow}$; let $\Zd(\theta)$ be the total weight of its spanning trees rooted at $1$---the downstream partner of $Z_r$, which is why it carries the same letter. It does not contain $u_j$, so it takes no drive argument. If a tree of $C_K$ uses the driven edge $p_j\to p_{j+1}$, then $p_j$ has spent its outgoing edge, so no downstream vertex can route through it: the downstream side is a rooted tree of $C_{j+1}^{\downarrow}$ and the upstream side a forest of $C_j$ rooted at $\{1,p_j\}$, chosen independently. Hence $w_\theta(\mathcal E)=u_jF_j\Zd$, and since $P_\theta$ and $R_\theta$ share the denominator $Z_K(\theta)$ in Eq.~\eqref{eq:PRdef}, it cancels: with $B_{K,j}$ from Lemma~\ref{lem:tail} and $F_j=g_jZ_j$,

\begin{equation}
\Gamma_\theta-1=\frac{P_\theta}{R_\theta}=\frac{B_{K,j}(\theta)}{u_jF_j(\theta)\,\Zd(\theta)}
=\frac{\prod_{r=j}^{K-1}v_r}{u_j\,g_j(\theta)\,\Zd(\theta)}.
\label{eq:Gminus}
\end{equation}
The prefactor in Eq.~\eqref{eq:Gminus} is free of $\theta$, so $Q_1=2R_\alpha R_1(P_\alpha/R_\alpha-P_1/R_1)\leq 0$ is exactly the statement that $g_j\Zd$ is nondecreasing. Its two factors pull in opposite directions---$g_j$ strictly decreases by Lemma~\ref{lem:mono}, $\Zd$ strictly increases---so nothing follows from them separately. Move one power of $\theta$ across the product:
\begin{equation}
  g_j(\theta)\,\Zd(\theta)
  =\underbrace{\bigl[\theta g_j(\theta)\bigr]}_{h_j(\theta)}
  \;\cdot\;
  \underbrace{\bigl[\Zd(\theta)/\theta\bigr]}_{\widetilde Z(\theta)} .
  \label{eq:split}
\end{equation}
Every rooted spanning tree of $C_{j+1}^{\downarrow}$ uses at least one fall-off edge---the root must have in-degree at least one, and the only edges into it are fall-offs---and the tree with $p_{j+1}\to1$ and $p_r\to p_{r-1}$ for $r>j+1$ uses exactly one, while the all-fall-off tree uses all $\nu$. So $\Zd$ is a polynomial in $\theta$ with nonnegative coefficients whose lowest and highest powers are exactly $\theta$ and $\theta^{\nu}$, and $\widetilde Z=\Zd/\theta$ is again a polynomial with nonnegative coefficients: nondecreasing, and constant if and only if $\nu=1$. Since $h_j$ is nondecreasing as well, the product in Eq.~\eqref{eq:split} is nondecreasing.

\app{app:optimum}{The optimal drive}%
From Eq.~\eqref{eq:master}, $\ln\epsilon=\ln\eeq+\ln\Phi_\alpha-\ln\Phi_1$ and
\begin{equation}
  \frac{\mathrm{d}\ln\Phi_\theta}{\mathrm{d}c}
  =\frac{R_\theta+P_\theta}{1+c(R_\theta+P_\theta)}-\frac{R_\theta}{1+cR_\theta}
  =\frac{P_\theta}{\mathcal D_\theta(c)},
  \label{eq:dlogPhi}
\end{equation}
with $\mathcal D_\theta=(1+cR_\theta)\bigl(1+c(R_\theta+P_\theta)\bigr)>0$. Hence $\mathrm{d}\ln\epsilon/\mathrm{d}c$ has the sign of $Q=P_\alpha\mathcal D_1-P_1\mathcal D_\alpha$. Expanding $\mathcal D_\theta=1+c(2R_\theta+P_\theta)+c^{2}R_\theta(R_\theta+P_\theta)$ and subtracting gives Eq.~\eqref{eq:Qcoef}, the products $P_\alpha P_1$ canceling in the linear coefficient. Lemma~\ref{lem:first} gives $Q_0<0$ and Appendix~\ref{app:Q1} gives $Q_1=2R_\alpha R_1(\Gamma_\alpha-\Gamma_1)\le0$.

If $Q_2>0$, $Q$ is an upward-opening parabola with $Q(0)<0$, so the product of its roots is $Q_0/Q_2<0$: exactly one root is positive, namely $c^*$ of Eq.~\eqref{eq:cstar}, with $Q<0$ before it and $Q>0$ after it. If $Q_2\le0$, every term of $Q(c)$ is nonpositive for $c>0$ while $Q_0<0$, so $\epsilon$ decreases throughout. For the value at the optimum, note that $\Phi_\theta=\mathcal D_\theta/(1+cR_\theta)^{2}$, so that
\begin{equation}
  \frac{\epsilon}{\eeq}=\frac{\Phi_\alpha}{\Phi_1}
  =\frac{\mathcal D_\alpha}{\mathcal D_1}
   \left(\frac{1+cR_1}{1+cR_\alpha}\right)^{2};
  \label{eq:valstep}
\end{equation}
at $c=c^*$ the stationarity condition $Q=0$ reads $\mathcal D_\alpha/\mathcal D_1=P_\alpha/P_1$, which turns Eq.~\eqref{eq:valstep} into Eq.~\eqref{eq:valstar}.

\app{app:sharp}{Sharpness}%
The bound given by Theorem~\ref{thm:main} is \emph{sharp}, as the following one-parameter family shows. Fix $K$, $j$, $\alpha$, recall $\nu=K-j$, and take $\delta\in(0,1)$,
\begin{equation}
  \begin{aligned}
    d_r&=1\ \ (1\le r\le K),\qquad u_j=\delta^{\,\nu+2},\\
    u_r&=v_r=1\ \ (r<j),\\
    v_r&=\delta\ \ (j\le r<K),\qquad u_r=\delta\ \ (j<r<K),
  \end{aligned}
  \label{eq:family}
\end{equation}
with the binding rates fixed by detailed balance, which is always possible. One number, $\delta$, sets every rate. Starting from equilibrium, choose $m_\delta=\delta^{-(\nu+1)}$, so that the driven energetic rate becomes $m_\delta u_j=\delta$: after driving, every chain rate from $p_j$ onward, forward and backward, equals $\delta$. Since $d_r=1$, Eq.~\eqref{eq:cf} gives $0<g_r(\theta;m_\delta)\le1/\theta$ for every $r\geq1$, and the first identity of Eq.~\eqref{eq:recursions} reads, for $r>j$,
\begin{equation}
\frac{Z_r(\theta;m_\delta)}{Z_{r-1}(\theta;m_\delta)}=\delta+\theta\bigl(1+\delta\,g_{r-1}(\theta;m_\delta)\bigr).
\label{eq:sharplayer}
\end{equation}
Therefore, exploiting the bounds given above on $g_{r-1}(\theta;m_\delta)$, we arrive at 
\begin{equation}
  \theta+\delta<\frac{Z_r(\theta;m_\delta)}{Z_{r-1}(\theta;m_\delta)}\le\theta+2\delta,
  \qquad r>j.
  \label{eq:sharpbound}
\end{equation}
Multiplying Eq.~\eqref{eq:sharpbound} over the $\nu$ layers $r=j+1,\dots,K$, and using the forced-tail identity $B_{K,j}(\theta)=\delta^{\nu}Z_j(\theta)$ of Lemma~\ref{lem:tail} gives
\begin{equation}
  \left(\frac{\delta}{\theta+2\delta}\right)^{\!\nu}
  \le\Pkj(\theta;m_\delta)
  <\left(\frac{\delta}{\theta+\delta}\right)^{\!\nu}.
  \label{eq:famest}
\end{equation}
Two consequences, both as $\delta\to0$. Taking the ratio of Eq.~\eqref{eq:famest}
at $\theta=\alpha$ and at $\theta=1$ traps
\begin{equation}
  \left(\frac{1+\delta}{\alpha+2\delta}\right)^{\!\nu}
  <\frac{\Pkj(\alpha;m_\delta)}{\Pkj(1;m_\delta)}
  <\left(\frac{1+2\delta}{\alpha+\delta}\right)^{\!\nu}
  \longrightarrow\alpha^{-\nu},
  \label{eq:famratio}
\end{equation}
while the lower bound alone gives, for $\theta\in\{1,\alpha\}$,
\begin{equation}
  (m_\delta-1)\,\Pkj(\theta;m_\delta)
  \ \ge\ \frac{\delta^{-1}-\delta^{\nu}}{(\theta+2\delta)^{\nu}}
  \ \longrightarrow\ \infty.
  \label{eq:famdiv}
\end{equation}
As before, write $a=\Pkj(\alpha;m_\delta)$, $b=\Pkj(1;m_\delta)$ and $c=m_\delta-1$, and factor $ca$ out of the numerator of Eq.~\eqref{eq:exacterror} and $cb$ out of its denominator: 
\begin{equation}
  \frac{\epsilon(m_\delta)}{\eeq}=\frac{1+ca}{1+cb}
  =\frac{a}{b}\cdot\frac{1+(ca)^{-1}}{1+(cb)^{-1}}\,.
  \label{eq:sharplimit}
\end{equation}
The first factor tends to $\alpha^{-\nu}$ by Eq.~\eqref{eq:famratio}, the second
to $1$ by Eq.~\eqref{eq:famdiv}. Hence $\epsilon(m_\delta)\to\eeq\,\alpha^{-\nu}=\eeq^{\,K-j+1}$. Theorem~\ref{thm:main} forbids equality at finite parameters, therefore,
\[
  \inf\epsilon=\eeq^{\,K-j+1},
\]
approached but not attained. $\square$

\smallskip\noindent
Figure~\ref{fig:tradeoff} shows this family. Driving it at the optimum $m^{*}$ of Theorem~\ref{thm:opt}---which exists here, since $Q_2>0$ for every entry of the table below---instead of at $m_\delta$ only helps, since $m^{*}$ minimizes $\epsilon$ over $m$: then $\eeq^{\,K-j+1}<\epsilon(m^{*})\le\epsilon(m_\delta)$, and the same limit follows by squeezing. Reducing $\delta$ drives the error onto the floor:
\begin{center}
\begin{tabular}{crrr}
\hline
$\delta$ & $j=1$ & $j=2$ & $j=3$\\
\hline
$10^{-1}$ & $2.6512$ & $1.6262$ & $1.2319$\\
$10^{-2}$ & $1.1211$ & $1.0565$ & $1.0236$\\
$10^{-3}$ & $1.0117$ & $1.0056$ & $1.0024$\\
$10^{-4}$ & $1.0012$ & $1.0006$ & $1.0002$\\
\hline
\end{tabular}
\end{center}
\noindent The entries are $\epsilon(m^{*})/\eeq^{\,K-j+1}$ for $K=4$, $\alpha=3$. Each decade of $\delta$ removes about a decade of the excess: the approach is $O(\delta)$, and the limit is the floor itself.

\end{document}